\documentclass[a4paper,USenglish,cleveref,autoref,thm-restate]{lipics-v2021}
\hideLIPIcs
\title{Hardness of High-Dimensional Linear Classification}

\titlerunning{Hardness of High-Dimensional Linear Classification}

\author{Alexander Munteanu}{TU Dortmund, Germany}{alexander.munteanu@tu-dortmund.de}{}{supported by the German Research Foundation (DFG) - grant MU 4662/2-1 (535889065) and by the TU Dortmund - Center for Data Science \& Simulation (DoDaS).}
\author{Simon Omlor}{TU Dortmund, Germany}{simon.omlor@tu-dortmund.de}{}{supported by the German Research Foundation (DFG) - project no. 535889065.}
\author{Jeff M. Phillips}{University of Utah, UT, USA}{jeffp@cs.utah.edu}{}{thanks his support from NSF CCF-2115677, 2421782, and Simons Foundation MPS-AI-00010515.}

\authorrunning{A. Munteanu, S. Omlor, J.~M. Phillips}

\Copyright{Alexander Munteanu, Simon Omlor, Jeff~M. Phillips}



\category{}

\relatedversion{}
\acknowledgements{We would like to thank Peyman Afshani for valuable discussion.}
\nolinenumbers

\usepackage[utf8]{inputenc}
\usepackage{lmodern}

\usepackage{mathtools}
\usepackage{amsmath,amsthm,amssymb}

\DeclareMathOperator*{\adj}{adj}

\usepackage{xspace}

\newtheorem*{fact}{Fact}

\newcommand{\R}{\mathbb{R}}
\newcommand{\Z}{\mathbb{Z}}
\renewcommand{\H}{\mathcal{H}}
\newcommand{\eps}{\varepsilon}
\newcommand{\MaxH}{{\rm \textsc{MaxHalfspace}}\xspace}
\newcommand{\sgn}{\mathrm{sign}}
\renewcommand{\det}{\ensuremath{\;\!\mathrm{det}}}
\newcommand{\poly}{\operatorname{poly}}

\begin{document}

\maketitle

\begin{abstract}
We establish new exponential in dimension lower bounds for the Maximum Halfspace Discrepancy problem, which models linear classification.
Both are fundamental problems in computational geometry and machine learning in their exact and approximate forms. 
However, only $O(n^d)$ and respectively $\tilde O(1/\varepsilon^d)$ upper bounds are known and complemented by polynomial lower bounds that do not support the exponential in dimension dependence.  We close this gap up to polylogarithmic terms by reduction from widely-believed hardness conjectures for Affine Degeneracy testing and $k$-Sum problems. Our reductions yield matching lower bounds of $\tilde\Omega(n^d)$ and respectively $\tilde\Omega(1/\varepsilon^d)$ based on Affine Degeneracy testing, and $\tilde\Omega(n^{d/2})$ and respectively $\tilde\Omega(1/\varepsilon^{d/2})$ conditioned on $k$-Sum.  
The first bound also holds unconditionally if the computational model is restricted to make \emph{sidedness} queries, which corresponds to a widely spread setting implemented and optimized in many contemporary algorithms and computing paradigms.
\end{abstract}

\allowdisplaybreaks

\section{Introduction}
\label{sec:intro}

\subparagraph{Linear Classification.}
One of the most fundamental tasks in machine learning is the minimum mislabeling linear classification problem.  We formulate it as follows.  Consider two point sets $R,B \subset \R^d$ (for red $R$ and blue $B$ points) so $|R\, \cup\, B| = n$.  Let $\H_d$ be the set of all halfspaces parameterized by $w,\xi \in \R^d \times \R$. So the \emph{halfspace} $h_{w,\xi} \in \H_d$ is the set of points $h_{w,\xi} := \{x \in \R^d \mid \langle w, x \rangle - \xi \geq 0\}$, where $\langle w, x \rangle = \sum_{j=1}^d w_j\cdot x_j$ denotes the standard Euclidean dot-product.   
Now given fixed input sets $R,B \subset \R^d$ and a halfspace $h \in \H_d$
define
\[
\phi(h) = \phi_{R,B}(h) = |h \cap R| - |h \cap B|. 
\]
as the number of red points in $h$, minus the number of blue points in $h$. Then the fraction of misclassifications may be written as
\[
\psi(h) = \frac{|R| - (|h \cap R| - |h \cap B|)}{|R \cup B|}
= 1 - \frac{|B|}{|R\cup B|} - \frac{1}{n}\phi(h).  
\]
So maximizing $\phi$ is equivalent to minimizing $\psi$.  For historical reasons we write the aim to minimize the misclassification rate $\psi(h)$ as the following maximization problem and its approximate version:
\begin{itemize}
    \item \MaxH:  For $R, B \subset \R^d$ with $|R \cup B| = n$, find $h^* = \arg\max_{h \in \H_d} \phi(h)$.
    \item $\eps$-\MaxH:  Consider $R, B \subset \R^d$ with $|R \cup B| = n$ with $h^* = \arg\max_{h \in \H_d} \phi(h)$.  For an error tolerance $\eps \in (0,1)$, find some $\hat h \in \H_d$ so that $\phi(h^*) - \phi(\hat h) \leq \eps n$.
\end{itemize}

While it is common in machine learning to replace $\psi$ with a convex loss function (such as logistic loss) which allows for a solution with gradient descent, the most classic formulations by Vapnik~\cite{vapnik-SLT,vapnik2015uniform}, which begat much of the foundations of statistical learning theory, have this combinatorial form and result in strong generalization guarantees.  

If $R \cup B$ is drawn i.i.d. from a distribution $\Pi$, then we can use $\psi(h)$ to predict how accurate $h$ would be on a new data point $x \sim \Pi$. In particular, for $\eps,\delta \in (0,1)$, if $|R \cup B| = \Theta(\eps^{-2}(d + \log(1/\delta)))$, then with probability at least $1-\delta$ over the sample, we get a guarantee that the probability that any $h$ correctly classifies a new point $x \sim \Pi$ is within $\eps$ to the finite sample fraction $\psi(h)$~\cite{vapnik2015uniform,li2001improved,ShalevSBD14}.
Thus maximizing $\phi(h)$ on a small sample gives a classifier which is within $\eps$ to the best possible generalization guarantee for new data.
Notably, this generalization rate guarantee is only accurate up to an additive $\eps \in (0,1)$. So solving for the exact optimum $h^* = \arg\max_{h \in \H_d} \phi(h)$ (as in \MaxH) may be unnecessary.  In such cases, it is sufficient to find some $\hat h$ satisfying $\eps$-\MaxH (and it can be randomized, succeeding with probability $1-\delta$), to achieve the same asymptotic guarantees.  

\subparagraph{Conditional Hardness.}
Conditional hardness results in computational geometry are often obtained by reduction from base problems that are believed to be hard. The main problems we reduce from are $k$-\textsc{Sum} and \textsc{AffineDegeneracy} testing. We discuss key variants and known results starting from the classic $3$-\textsc{Sum}~\cite{GajentaanO95}.

Consider the \underline{$k$-\textsc{Sum} decision problem}: Given $n\geq k$ integers $a_i \in \mathbb{Z}$, decide whether there exists a (multi)set $S\subseteq \{a_1,\ldots,a_n\}$ of $|S|=k$ integers that add up to $0$.

The problem can be solved in time $O(n^{\lceil k/2 \rceil})$ by comparing $k/2$-tuples from one half of the input to the other half, known as the \emph{meet in the middle} approach~\cite{horowitz1974computing}.
Despite polylogarithmic improvements for constant $k$ \cite{BaranDP08}, $\Omega(n^{\lceil k/2 \rceil})$ lower bounds hold against algorithms in certain restricted models of computation \cite{Erickson95,AilonC05} and in general the problem is conjectured to not admit $O(n^{\lceil k/2\rceil-c})$ time algorithms for any $c>0$, cf. \cite{AbboudL13}.

\begin{conjecture}[$k$-{Sum} conjecture]
For $k\geq 2$, there does not exist a randomized algorithm that succeeds (with high probability) in solving $k$-{\rm \textsc{Sum}} in time $O(n^{\lceil k/2\rceil-c})$ for any $c > 0$.
\end{conjecture}

Consider the \underline{\textsc{AffineDegeneracy} problem}: Given $n$ points in $\mathbb{Z}^d$, decide whether there exists a set of $(d+1)$ points that lie on a common affine hyperplane in $d$ dimensions.

If the algorithm is restricted to \emph{sidedness} queries that return whether a given point lies above, below, or on a given affine hyperplane, then Erickson and Seidel~\cite{EricksonS95} prove a $\Omega(n^d)$ lower bound matching the $O(n^d)$ complexity of any existing algorithm for the problem \cite{ChazelleGL85,EdelsbrunnerOS86,EdelsbrunnerG89,EdelsbrunnerSS93}. 
Recent research by Cardinal and Sharir~\cite{CardinalS25} has explored lower-order improvements in the exponent to $O(n^{k-1+o(1)})$ in the case of constant-degree polynomials in one dimension. Their result thus indicates that lifting the polynomial degree is slightly simpler than lifting to higher dimensions.
The general case of \textsc{AffineDegeneracy} in $d$ dimensions remains a major open problem in the real RAM model, for which it is conjectured that there exists no algorithm that runs in time $O(n^{d-c})$ for any $c>0$, cf. \cite{Har-PeledIM18,chan2022hopcroft,CardinalS25}.

\begin{conjecture}[{Affine Degeneracy} conjecture]\label{con:affinedegen}
There does not exist an algorithm that solves {\rm\textsc{AffineDegeneracy}} in $d$ dimensions in time $O(n^{d-c})$ for any $c > 0$.
\end{conjecture}

\subsection{Our Results}
We provide hardness results for exact and approximate versions of the \MaxH problem that hold conditionally on the aforementioned widely believed hardness conjectures in the real RAM model of computation and unconditionally in a more restricted computational model that counts only $\sgn(x^Ty)\in\{-1,0,+1\}$ queries, referred to as \emph{sidedness} query model or (ternary) $k$-linear decision tree model in the literature.
For an overview, we only state our bounds slightly informally for the case $d$ is a constant, and so for instance $2(d+1)^d$ is also a constant. The exact bounds are stated in the formal part in~\Cref{sec:formalpart}.

\subparagraph{\texorpdfstring{$k$}{k}-{Sum} Hardness.}
Our first reduction is from $k$-\textsc{Sum} to \textsc{MaxDiscrepancy}. This generalizes the previous $3$-\textsc{Sum} reduction \cite{GajentaanO95,MathenyP21} that gave quadratic lower bounds to an arbitrary $k\geq 2$ and hereby improves the exponents in the lower bounds to depend on the dimension.
Our reduction yields the following theorem:

\begin{theorem}[informal version of \Cref{thm:ksum}]
Assume that there exists no algorithm that solves $k$-{\rm \textsc{Sum}} on $n$ items within runtime $O(n^{\lceil k/2 \rceil - c})$ for any $c > 0$. Then there exist sets $R,B \subset \R^d, d=k-1$, each of size $n/2$, for which there is no algorithm that solves either 
\begin{itemize}
   \item \MaxH in time $O(n^{\lceil (d+1)/2 \rceil - c})$, or
   \item $\eps$-\MaxH in time $O(1/\eps^{\lceil (d+1)/2 \rceil - c} )$.
\end{itemize}
\end{theorem}

\subparagraph{Affine Degeneracy Testing Hardness.}
Our second reduction is from \textsc{AffineDegeneracy} testing to \textsc{MaxDiscrepancy} and yields the following theorems:

\begin{theorem}[informal version of \Cref{thm:affinedegeneracy}]
Assume that there exists no algorithm for {\rm \textsc{AffineDegeneracy}} on $n$ points with runtime $O(n^{d-c})$ for any $c >0$. Then there exist sets $R,B \subset \R^d$, each of size $n/2$, for which there is no algorithm that solves either 
\begin{itemize}
    \item \MaxH in time $O(n^{d-c})$, or
    \item $\eps$-\MaxH in time $O(1/\eps^{d-c})$.
\end{itemize}
\end{theorem}

\begin{theorem}
There exist sets $R, B \subset \R^d$, each of size $n/2$, for which any algorithm that uses only sidedness queries requires
\begin{itemize}
    \item $\Omega(n^d)$ such queries to solve \MaxH, and 
    \item $\Omega(1/\eps^d)$ such queries to solve $\eps$-\MaxH.  
\end{itemize}
\end{theorem}

We remark that lower bounds of the form $1/\eps^d$ do not rule out $2^{1/\eps}\poly(d)$ algorithms.
The latter two results are of special importance since either under the hardness conjecture or in the sidedness query model they match the best exact and approximate algorithms for this problem up to lower order polylogarithmic terms. The exact version, \MaxH, has a $O(n^d)$ complexity \cite{DobkinE93,DobkinEM96}, and for the approximate version, $\eps$-\MaxH, there is an $O(\frac{1}{\eps^d}\log^4 \frac{1}{\eps})$ time algorithm (after preprocessing the input in $O(n)$ time) in \cite{MathenyP21}. Notably, this is based on a range counting data structure that can be implemented to work in the sidedness query model as well.
The sidedness query model may seem restrictive at first glance and indeed it does not rule out simple comparisons between values (rather than signs) of dot products and determinants which may be exploited to break such lower bounds.

However, we stress that basically all known algorithms rely exactly on that type of subroutines, see related work below, and especially the discussion in~\cite{EricksonS95}.
Now, 30 years later, the core sidedness query operation is precisely the heavily optimized $\sgn(x^Ty)$ operation which can be vectorized, optimized and accelerated on GPUs. Additionally, the study of big data models have brought up a paradigm that discretizes and reduces the search space to a few candidate solutions and then brute forces that subset using simple queries as in \cite{MathenyP21}. This underlines that the sidedness query model is even more relevant in today's age of AI/ML and big data, than it probably was when Erickson and Seidel~\cite{EricksonS95} studied it.

The hardness conjecture of \textsc{AffineDegeneracy} extends the same lower bound to the real RAM model in any dimension. In the Turing model, our reduction holds only in any constant dimension, due to the bit complexity of representing and adjusting small scalars.

\subsubsection{Technical Overview}

\paragraph*{Reduction from $k$-{Sum}}

Given an input $a_1,\ldots,a_n \in\Z$ for the $k$-\textsc{Sum} decision problem, the idea behind our construction is as follows: we construct for any $i\in [n]$ and $j\in[k]$ a point $x_{i,j} \in \R^{k-1}$. The points are constructed in such a way, that\footnote{up to technical scalars that are specified in the formal proof, see \Cref{sec:ksumhardness}.} the first coordinate of each $x_{i,j}$ corresponds to the input $a_i$, and for $j \in 
[1,\ldots, d-1]$, the remaining coordinates are the standard basis vectors $e_j$, whose coordinates are all $0$ except for coordinate $j$ which is $1$.
The remaining vectors are $x_{i,0}$ and $x_{i,d}$.  Both have their first coordinate set to $-a_i$ (times a constant), then $x_{i,0}$ has all remaining coordinates set to $0$, and $x_{i,d}$ has all remaining coordinates are set to $1$. 

The idea behind this construction is that a solution $a_1,\ldots,a_k$ to $k$-\textsc{Sum} corresponds to a hyperplane that crosses a specific selection of points, including an ancillary point that acts as an intercept term. Hence, there exists a linear combination of these points such that their first coordinates resemble the equation $a_1 + \ldots + a_{k-1} = -a_k$, while the other coordinates enforce constraints that ensure we make the right geometric selection of points, such that on these coordinates we have $e_1 + e_2 + \ldots + e_{k-1} = \boldsymbol{1}$, i.e., the selected unit basis vectors sum up to the vector $\boldsymbol{1}$ that comprises only $1$s. By construction, a linear combination of these points whose coefficients parametrize a hyperplane also implies the existence of a $k$-\textsc{Sum} solution.

Now, we reduce to the \MaxH problem. To this end, we copy each of the constructed points twice to create $nk$ red points $r_{i,j}=x_{i,j}+\gamma e_1\in R$, and another $nk$ blue points $b_{i,j}=x_{i,j}-\gamma e_1\in B$. That means the two colored point sets to be separated are simply copies of the previously constructed points but shifted by a tiny amount $\gamma$ in opposite directions along their first coordinate. Note that this affects only the coordinate that corresponds to the original $k$-\textsc{Sum} input, not the additional gadgets.

Now if $||h\cap R|-|h\cap B||\geq k$ then there exists a hyperplane that passes through the tiny gaps between at least $k$ pairs of points $(r_{i,j}, b_{i,j})$. Using the constraints imposed by the gadgets this enforces a certain parameterization of the hyperplane. Applying this parameterization only to the first coordinates of the points and choosing $\gamma \leq 1/k$, we can bound $|a_1+\ldots+a_k| \leq (k-1)\gamma < 1$. Since all $a_{i}$ are integers, it follows that $a_{1} + \ldots + a_{k} = 0$.

For the other direction $a_{1} + \ldots + a_{k-1} = - a_{k}$ implies, as explained above, that there exists a linear combination of $k$ original points $x_{i,j}$ that equals zero. Now note that the shifted versions $r_{i,j}, b_{i,j}$ of these $k$ points each lie above (respectively below) the hyperplane, so each of them contributes $1$ to the cost. Other points can only increase this number above $k$, or their red and blue contributions cancel in case they both lie on the same side. It thus follows that $||h\cap R|-|h\cap B||\geq k$.

\paragraph*{Reduction from {AffineDegeneracy}}
This construction is more direct: given $n$ points $x_i\in\Z^d$ as input for the \textsc{AffineDegeneracy} testing problem, we simply create $n$ red points $r_i=x_i+\gamma e_1\in R$, and another $n$ blue points $b_i=x_i-\gamma e_1\in B$. The two colored point sets to be separated are simply copies of the original points shifted by a tiny amount $\gamma$ in opposite directions along their first coordinate.

If there exist $d+1$ points that lie on a common affine hyperplane, then the number of red points minus the number of blue points of the transformed input that lie in the halfspace $h$ is clearly at least $||h\cap R|-|h\cap B||\geq d+1$. To see this, note that for any of the original points on the hyperplane, we have one red point that lies above, and one blue point that lies below. Other points can only increase this number, or their red and blue contributions cancel if both lie on the same side.

The other direction is considerably more sophisticated and technical to prove: if $||h\cap R|-|h\cap B||\geq d+1$ then there exist at least $d+1$ pairs $(r_i, b_i)$ of corresponding red and blue points, such that the hyperplane passes through some points $x'_i$ located in the tiny gap of size $2\gamma$ between $r_i$ and $b_i$. Relabel these points $x'_1,\ldots,x'_{d+1}$. Our claim is that the corresponding $d+1$ original points $x_1,\ldots,x_{d+1}$ lie on a common affine hyperplane. Since any collection of $d$ points such as $x_1,\ldots,x_{d}$ lie on a common hyperplane, the task reduces to show that $x_{d+1}$ also lies on the same hyperplane.

Reframing this in terms of linear algebra, we collect the vectors $x'_1,\ldots,x'_{d}$ into a matrix $X'\in\R^{d\times d}$, and similarly $x_1,\ldots,x_{d}$ into $X\in\Z^{d\times d}$. Let $\mathbf{1}\in\R^d$ be the vector whose coordinates are all equal to one. We can now calculate the vector of coefficients of the hyperplane as $b' = X'^{-1}\mathbf{1}$. And for this instance, we know that $\langle b', x'_{d+1}\rangle-1 = 0$. Similarly, we get for the original points $b = X^{-1}\mathbf{1}$, and our task is to prove that $\langle b, x_{d+1}\rangle-1 = 0$ holds as well. To this end, we prove the following two claims:
\begin{enumerate}
    \item[(i)] we have that $\langle b,x_{d+1} \rangle- 1$ is an integer multiple of $1/N_0$, for some integer $N_0\neq 0$, and
    \item[(ii)] it holds that
    \[
        |\langle b,x_{d+1} \rangle- 1|=|\langle b,x_{d+1} \rangle- \langle b',x'_{d+1} \rangle + \langle b',x'_{d+1} \rangle - 1 | =|\langle b,x_{d+1} \rangle-\langle b',x'_{d+1} \rangle| < 1/N_0.
    \]
\end{enumerate}
    Combining claims (i) and (ii) clearly implies our goal $\langle b,x_{d+1} \rangle-1=0$.

At this point of the proof, especially bounding $|\langle b,x_{d+1} \rangle-\langle b',x'_{d+1} \rangle| < 1/N_0$ becomes very technical and requires bounding the effect of  the remainder of matrix inverses $(X^{-1} - X'^{-1})$ along the shift direction $e_1$, as well as along $x_{d+1}$. Crucial steps involve showing that the determinants of $X$ and its minors must be non-zero integers. Then we note that $X'$ can be seen as a rank-one update of $X$, and apply the Sherman-Morrison formula to express the required remainder term in closed form. We further use Hadamard's bound on the determinants of minors of $X$ in terms of the largest integer that appears in $X$. Choosing the shift size parameter $\gamma$ small enough to cancel all but a factor $1/\det(X)$ yields the bounds that we require to conclude the proof. We refer the reader to \Cref{sec:affinedegenhardness} for the full details.

\subsection{Implications Beyond Machine Learning}   
A number of other settings consider range space maximization, where the key unit of measure is the fraction of red points $\mu_R(h) = \frac{|h \cap R|}{|R|}$ and of blue points $\mu_B(h) = \frac{|h \cap B|}{|B|}$ in a halfspace.  

\begin{itemize}
   \item \textbf{discrepancy:} $\phi_\|(h) = |\mu_R(h) - \mu_B(h)|$    \\ 
      This relates to discrepancy theory~\cite{matousek1999geometric,chazelle2000discrepancy,bansal2010constructive,ahmed2024communication} where one often seeks to find a coloring of a set $X \subset \R^d$ so each $x \in X$ is assigned to $R$ or to $B$, and so $\max_{h \in \H}\phi_\|(h)$ is minimized over any choice of $R,B$.  Then taking $R,B$ as input, the optimization $\max_{h \in \mathcal{H}}\phi_{\|}(h)$ amounts to evaluating how well the chosen coloring works.  In particular, many of the algorithms are Monte Carlo randomized, so to transform them into a Las Vegas style algorithm this verification step would be required.  Moreover, iteratively applying discrepancy leads to coresets, which achieve additive $\eps |X|$ error on $\phi$ or range counting queries; hence the approximate version is often sufficient.  

      In minimizing discrepancy, one often seeks $|R|=|B|$, in which case $\phi_\|(h) = \frac{2}{n} | \phi(h) |$.  Further note that the absolute value $| \cdot |$ takes the maximum value between $\phi(h)$ and when reversing the roles of $R$ and $B$. Our constructions will be symmetric in $R$ and $B$ with $|R|=|B|$. So our hardness results will hold for the most common case in discrepancy maximization in high dimensions.

    \item \textbf{Poisson:} $\phi_P(h) = \mu_R(h) \log \frac{\mu_R(h)}{\mu_B(h)} + (1-\mu_R(h)) \log \frac{1-\mu_R(h)}{1-\mu_B(h)}$ \\
       This relates to the spatial scan statistic model under a Poisson point process, as described by \cite{kulldorff1997spatial}.  The maximizing $h \in \H$ corresponds to the most anomalous region under a statistical hypothesis test where the null hypothesis $\texttt{H}_0$ models that there is a consistent Poisson process for all of the data, and the $\texttt{H}_1$ hypothesis models that the Poisson rate is different in $h$ than outside $h$. Then $\phi_P$ is the negative log-likelihood ratio of these two scenarios, also known as the KL divergence between the two distributions $(\mu_R(h), 1-\mu_R(h))$ and $(\mu_B(h), 1-\mu_B(h))$.   
       The score $\phi_P(h^*)$ of the maximizing $h^* = \arg\max_{h \in \H} \phi_P(h)$ is the individually most powerful statistic for evaluating this sort of anomaly.  Moreover, efficient $\eps$-approximate algorithms~\cite{matheny2016scalable} retain high statistical power.  

       One can cover the relevant part of $\phi_P$ with $O(\frac{1}{\eps} \log \frac{1}{\eps})$ linear functions $\phi_\alpha(h) = \alpha |h \cap R| + (1-\alpha) |h \cap B|$, so returning the maximum $h$ over any function $\eps$-approximately maximizes $\phi_P$.  And by changing the weight of points in $R$ and $B$ to $\alpha$ and $(1-\alpha)$ our results provide hardness for any of these linear instances.

       In spatial scan statistics, it is common to search over balls \cite{kulldorff1997spatial,matheny2016scalable} instead of halfspaces (although halfspaces are considered~\cite{matheny2020scalable}), and procedurally, a Veronese map from balls in $\R^{d-1}$ to halfspaces in $\R^d$ reduces this to a halfspace problem.  All known methods with provable runtime and error guarantees~\cite{matheny2016scalable,MathenyP21} either explicitly use or are isometric to using this transformation.   

\end{itemize}

\subsection{Additional Related Work}
\subparagraph{Maximum Discrepancy Classification.}
Exact algorithms are known that compute halfspace discrepancy in $O(n^d)$ time in any dimension $d\geq 2$ \cite{DobkinE93,DobkinEM96,dobkin1996computing}.

It was shown in \cite{MathenyP21} that the $\eps$-\MaxH problem can be solved in $O(\frac{1}{\eps^d}\log^4\frac{1}{\eps})$ time (plus a linear scan over the data in time $O(n)$). This was complemented by conditional lower bounds in $2$ dimensions: an exact solution requires $\Omega(n^{3/2})$ time conditioned on the conjecture that the all-pairs-shortest-paths problem (\textsc{APSP}) requires $\Omega(n^3)$ time to be solved. An approximate solution requires $\Omega(1/\eps^{2})$ by reduction from $3$-\textsc{Sum} under the special case $k=3$ of the aforementioned $k$-\textsc{Sum} conjecture detailed in \Cref{sec:intro}.

Previous results by the same authors \cite{MathenyP18} gave $O(\frac{1}{\eps^{d+1/3}}\log^{2/3}\frac{1}{\eps})$ for $\eps$-\MaxH. They also studied a version that considered rectangular ranges instead of halfspaces, and provided an $O(n + \frac{1}{\eps^2} \log \log \frac{1}{\eps})$ algorithm and an $\Omega(1/\eps^2)$ lower bound leveraging a previous $\Omega(n^2)$ lower bound of~\cite{BackursDT16} for the exact rectangle variant. These lower bounds were again conditioned on the conjecture that \textsc{APSP} requires $\Omega(n^3)$ time to be solved.

\subparagraph{More Connections to Learning Theory.}
Aside for the classic work described above, the hardness of classification has been studied in computational learning theory from other perspectives.  A touchstone result by Guruswami and Raghvendra~\cite{guruswami2009hardness} showed it was NP-Hard in the dimension $d$ to distinguish between instances where there exist halfspaces which agree on at least a $(1-\eps)$ fraction of the labels and ones where no halfspace agrees on more than a $(1/2+\xi)$ fraction of the data. Their work, and many others in this line (c.f. \cite{feldman2012agnostic}), consider data from $\{0,1\}^d$, and since $O(d/\eps^2)$ samples are sufficient to attain this $\eps$-error, the number of points $n$ does not show up separately in this analysis. But this work does not for instance control the constant in the exponent of the runtime; it does not rule out algorithms with runtime $1/\eps^{\max\{2, d/1{,}000{,}000\}}$.  
Also, algorithms with runtime polynomial in $1/\eps$ and $d$ are possible with stronger noise like Massart noise~\cite{diakonikolas2019distribution}.  

A very recent and independent work by Pinto, Palit, and Raskhodnikova~\cite{pinto2026computational} approaches these problems from the property testing perspective, and conditioned on $k$-\textsc{Sum}, attains lower bound on the runtime for $\eps$-\MaxH of $\Omega(1/\eps^{\lceil (d+1)/2 \rceil - c} )$.  The core of their argument and construction is basically the same as our $k$-\textsc{Sum} hardness result.  

\subparagraph{No-Dimensional Results in Computational Geometry and Classification.}
Computational geometry has recently witnessed a number of ``no-dimensional'' or ``dimension-free'' results. These are often approximate solutions to problems in $\R^d$ where the core aspect of the problem does not depend on the dimension $d$. It may take $O(d)$ time to read or process each data point, so they are not independent of the dimension $d$, but importantly, these sorts of results prevent exponential dependence. Arguably this expedition started with approximate minimum enclosing ball and geometric median problems~\cite{Badoiu02,BadoiuC03,KrivosijaM19}, which only require $O(1/\eps)$ resp. $O(1/\eps^2)$ points to approximate within a relative $(1+\eps)$. This sort of analysis led to more general analyses of Frank-Wolfe style greedy algorithms~\cite{clarkson2010coresets}, which induced no-dimensional results for the approximate polytope distance problem~\cite{gartner2009coresets}.  
Also, classical combinatorial geometric consequences of convexity such as Carath\'eodory, Helly, and Tverberg theorems can be made no-dimensional when the diameter is bounded~\cite{adiprasito2020theorems}. And this led to improved centerpoint and $\eps$-net results \cite{choudhary2022no,har2025no}. 

These geometric insights have led to improvements in sublinear and no-dimensional algorithms for classic linear classification problems~\cite{clarkson2012sublinear,ding2021stability}.  These analyses, and the above mentioned analysis of polytope distance~\cite{gartner2009coresets}, are refinements of Novikoff's classic analysis of the perceptron for linear classifiers, which is also no-dimensional.  In particular, it requires that there exists a perfect classifier with no mistakes, and depends on the ratio between the distance between classes, and the diameter of the input set, cf. \cite{Har-PeledRZ07,Har-Peled15}. The linear classification problem which optimizes an often convex loss function (e.g., logistic regression) requires additional assumptions to be made no-dimensional. Classic sparsity assumptions do not prevent from mild $\Theta(\log(d)/\eps^2)$ dependence~\cite{MaiMMRSW23}. Strictly dimension-free bounds can be achieved if one uses a regularization term and also assumes the diameter is bounded (in expectation)~\cite{AlishahiP24}. These results $(1+\eps)$-approximate typical classification loss functions with $O(1/\eps^2)$ samples and an additional dependence that relies solely on the amount and type of regularization.
Apart from classification there is a plethora of no-dimensional results in the clustering regime. The first dimension-free bounds for $k$-means/median clustering objectives appeared in \cite{FeldmanSS20,SohlerW18}. Recent developments extend the classic finite data setting to learning-theoretical dimension-free bounds for center based clustering and classification of points and polygonal curves~\cite{BucarelliLST23,Krivosija25,AlishahiP24}.

\subparagraph{Impossibility of No-Dimensional Results.}
In the context of these numerous no-dimensional results, it may be surprising that the geometric linear classification problem, and especially the approximately optimized variant tackled in this paper, is shown to require time exponential in $d$ under conditional hardness assumptions.  What is different in our setting?  
First, the perceptron-based results require separability, so $\psi(h^*) = 0$; our constructions do not consider this case, and have $\psi(h^*)$ fairly large.  
Second, these results mostly consider a bounded diameter of the point set, and also a structural term that prevents affine scaling, which could be the margin or the regularization term that has a similar effect. While we could affinely scale our constructions to be within a unit ball, it would make these other structural parameters (e.g., the margin between the properly separated points) proportionally small.  
Third, while no-dimensional $\eps$-nets~\cite{har2025no} and $\eps$-covers~\cite{phillips2025dimension} exist -- and these play a crucial role in the best algorithms for our variant of classification~\cite{MathenyP18,MathenyP21} -- these results apply for bounded-radius balls or kernels, and not the halfspaces we study.  

Hence, we hope our conditional hardness results help the community identify when these high-dimensional geometric problems relevant for data analysis can, \emph{and cannot}, be solved in a dimension-efficient way. Our work provides recipes to show that even when approximation is allowed, exponential dependence on the dimension can be unavoidable.

\section{Improved Halfspace Discrepancy Lower Bounds}\label{sec:formalpart}

\subsection{\texorpdfstring{$k$}{k}-{Sum} Hardness}\label{sec:ksumhardness}

We reduce from the $k$-\textsc{Sum} decision problem.

\begin{theorem}\label{thm:ksum}
Assume that there exists no algorithm that solves $k$-{\rm \textsc{Sum}} on $n$ items within runtime $O(n^{\lceil k/2 \rceil - c})$ for any $c > 0$. Then there exist sets $R,B \subset \R^d$, for $d=k-1$, each of size $n/2$, for which there is no algorithm that solves either 
\begin{itemize}
   \item \MaxH in time $O((n/(2(d+1))^{\lceil (d+1)/2 \rceil - c})$, or
   \item $\eps$-\MaxH in time $O( (1/(2(d+1)\eps))^{\lceil (d+1)/2 \rceil - c} )$.
\end{itemize}
\end{theorem}

\begin{proof}
We reduce $k$-\textsc{Sum} to \textsc{MaxHalfspace} with $d=k-1$ as follows. 
Let $A = \{a_1, \ldots, a_n\}$ for integers $a_i \in \mathbb{Z}$ be an instance of $k$-\textsc{Sum}.  
We will transform this input into a set of $2nk$ points in $\mathbb{R}^d$ where the solution to \textsc{MaxHalfspace} will resolve the $k$-\textsc{Sum} instance on $A$.  

\subparagraph*{The Construction.}
For any integer $a_i, i\in[n]$ consider the following $k=d+1$ points in $\mathbb{R}^d$:
Let $z_{i,0}:=(-a_i/(2d-3), 0, \ldots, 0)$ and for each $j \in [d-1]$, let $z_{i, j}=(a_i/2, 0, \ldots, 0, 1, 0 \ldots, 0)$ which has a single $1$ entry at position $j+1$.
Further we create the point $z_{i, d}=(-a_i, 2, 2, \ldots, 2) $.

Further for each point $z_{i, j}$ we create two points: one point, which we place in $R$ and one point, which we place in $B$.  For $j \in \{0,\ldots,d\}$ we define $x_{i, j} \in R $ by $x_{i, j}= z_{i, j}+ \gamma e_1 $ and $y_{i, j} \in B $ by $y_{i, j}= z_{i, j} - \gamma e_1 $ where $\gamma$ is a sufficiently small number (for instance $\gamma < 1/(4d)$) and $e_1$ denotes the first standard basis vector in $\mathbb{R}^d$.

\subparagraph*{Restricting $R-B$ Difference to at least $k$.}
First, assume that there is $h$ with $\mu_R |R|-\mu_B |B|\geq k $. 
Consider a set of $k=d+1$ lines defined by $l_0(t)=(t, 0 , \ldots, 0) $, $l_j(t)=(t, 0, \ldots, 0, 1, 0, \ldots, 0)$, for $j\in[d-1]$, and $l_d(t)=(t, 2, \ldots, 2)$. Observe that all points in $R$ and $B$ lie on these lines.  
Let $v_j$ be the point where $h$ crosses the line $l_j$, for $j\in \{0,\ldots,d\}$.

Note, since the $z_{i,j}$ points are further apart than $2\gamma$, then the number of points $x_{i,j}$ or $y_{i,j}$ on each of the lines that lie above $h$ in $R$ can differ by at most $1$ from the number of points above $h$ in $B$.

Thus, there cannot be any hyperplane with $\mu_R |R|-\mu_B |B|>k$ and in order to achieve $\mu_R |R|-\mu_B |B|=k$, the hyperplane $h$ must cross each line $l_j$ at a point $v_{j}=(\beta_{j+1}, \ldots)$ that for some $i_j\in[n]$, lies between the points $ z_{i_j, j}- \gamma e_1$ and $ z_{i_j, j}+ \gamma e_1$. Equivalently, we have that $\beta_{j+1} \in [e_1^Tz_{i,j}- \gamma, e_1^Tz_{i,j}+\gamma] $. To simplify notation, we define $\alpha_{i,j}\coloneqq e_1^Tz_{i,j-1}$ for all $(i,j)\in [n]\times [k]$, where $k=d+1$. With this substitution, we get that $\beta_{j} \in [\alpha_{i,j}- \gamma, \alpha_{i,j}+\gamma] $.

\subparagraph*{Showing $k$-Separation implies $k$-{Sum}.}
Consider the points $w_j = v_j-v_0$, for $j\in [d]$. Since these are $d$ points in $\mathbb{R}^d$, they must lie on a common linear hyperplane and the only way to align the upper $d-1$ coordinates is by setting $w_d = 2\sum_{j=1}^{d-1} w_j$.

This is equivalent to $0=v_0 - v_d + 2 \sum_{j=1}^{d-1} (v_j-v_0) = -(2d-3) v_0 + (2 \sum_{j=1}^{d-1} v_j) - v_d$,
implying in particular for the first coordinate that
\begin{equation}\label{eqn:bzero}
    -(2d-3) \beta_1+ \left(2\sum\limits_{j=2}^{k-1} \beta_j\right) - \beta_{k}= 0 .
\end{equation}
Now, since for all $j\in[k]$ there exists some $i_j\in[n]$ such that  $|\alpha_{i_{j},j}-\beta_j|\leq \gamma$, we get that 
\begin{align*}
    & |a_{i_1} + \ldots + a_{i_k}| \\
    &\overset{Eq.~\eqref{eqn:bzero}}{=} \left|a_{i_1} + \ldots + a_{i_k} -
    \left(-(2d-3) \beta_1+ \left(2\sum_{j=2}^{k-1} \beta_j\right) - \beta_{k}\right)\right|\\
    &\quad=\quad\!\left|-(2d-3)\alpha_{i_1, 1} + \left(2 \sum_{j=2}^{k-1} \alpha_{i_j, j} \right) - \alpha_{i_k, k} -
    \left(-(2d-3) \beta_1+ \left(2\sum_{j=2}^{k-1} \beta_j\right) - \beta_{k}\right)\right|\\
    &\quad\leq\quad (2d-3) |\alpha_{i_1,1}-\beta_1| + \left( \sum_{j=2}^{k-1} 2 |\alpha_{i_j,j}-\beta_j| \right) + |\alpha_{i_k,k}-\beta_k| \leq 4d\gamma < 1\,,
\end{align*}
where we used the triangle inequality and the choice $\gamma < 1/(4d)$.
Since all $a_{i}$ are integers, it follows that $a_{i_1} + \ldots + a_{i_k} = 0$.

\subparagraph*{Showing $k$-{Sum} implies $k$-Separation.}
For the other direction, assume that $a_{i_1}+ \ldots + a_{i_k}=0 $ or equivalently $a_{i_1}+ \ldots + a_{i_{k-1}}= -a_{i_k}$.
Then the points $z_{i_1, 0}, \ldots, z_{i_k, d} $ lie on a hyperplane $h$. To prove this, note that any set of $d$ points in $d$ dimensions lie on a hyperplane. We need to show that the remaining point $z_{i_k, d}$ lies on the same hyperplane as the others. By a similar calculation as above, we get
\begin{align*}
    z_{i_k, d}=(-a_{i_k}, 2, \ldots , 2)
    &= (a_{i_1}+ \ldots + a_{i_{k-1}}, 2 ,\ldots, 2) \\
    &=-(2d-3)z_{i_1, 0}+2 \sum_{j=1}^{d-1} z_{{i_{j+1}}, j}
    = z_{i_1, 0}+ 2 \sum_{j=1}^{d-1} (z_{i_{j+1}, j}- z_{i_1, 0}).
\end{align*}
For this hyperplane it holds that $|\mu_R |R|-\mu_B |B|| \geq k$, since for any point $z_{i, l}$, for $l\in \{0,\ldots, d\}$ that lies on the hyperplane, the corresponding $x_{i, l}$ (and $y_{i, l}$) lie above (respectively below) the hyperplane, all counting for one of the two classes. If it does not lie on the hyperplane, then by construction and for sufficiently small $\gamma$ as above, both $x_{i, l}$, and $y_{i, l}$ lie on the same side of the hyperplane and thus both contribute $0$ or cancel each other.

\subparagraph*{Exact Bound.}
The number of points is $2nk=2n(d+1)$, and so any algorithm with runtime $O((\beta n)^\alpha)$ for \MaxH yields an algorithm with runtime $O((\beta n)^\alpha)$ for $k$-\textsc{Sum}.  Set $\beta = 1/(2(d+1))$ and $\alpha=\lceil k/2 \rceil - c= \lceil (d+1)/2 \rceil - c$.   Now using the Lemma's assumption on $k$-\textsc{Sum} hardness, there cannot be an algorithm for \MaxH within runtime $O((\beta n)^\alpha) = O((n/(2(d+1))^{\lceil (d+1)/2 \rceil - c})$.

\subparagraph*{Cost of Reduction with Error $\eps$.}
Similar to the reduction for \MaxH, we now set $\varepsilon=1/(2n(d+1))$, and any algorithm with runtime $O((\beta/\eps)^\alpha)$ for $\eps$-\textsc{MaxHalfspace} yields an algorithm with runtime $O((\beta/\eps)^\alpha)=O( 2n(d+1)\beta)^\alpha)$ for $k$-\textsc{Sum}.
We again set $\beta=1/(2(d+1))$ and $\alpha=\lceil k/2 \rceil - c= \lceil (d+1)/2 \rceil - c.$
Thus, assuming the $k$-\textsc{Sum} hardness conjecture as in the statement of the lemma, there cannot be any algorithm that solves $\varepsilon$-\textsc{MaxHalfspace} within runtime $O( 1/(2(d+1)\varepsilon)^{\lceil (d+1)/2 \rceil - c} )$.
We note that the time for the reduction from $k$-\textsc{Sum} to \textsc{MaxHalfspace} is $O(n k^2)\subseteq O(n^\alpha)$. 
\end{proof}

\subsection{Affine Degeneracy Hardness}\label{sec:affinedegenhardness}
The next lower bound is derived by reduction from the \textsc{AffineDegeneracy} problem. We first need a review of some basic linear algebra facts that are used in our proof.

\subsubsection{Preliminaries and Review of Linear Algebra}
Let $[d]=\{1,\ldots,d\}$ denote the positive integers up to $d$. Let $\mathbb{Z}_N = \{-N, \ldots, N\}$ be the set of integers between $-N$ and $N$. Given a square $d\times d$ matrix $X$, we denote its $i,j$-th entry for $i,j\in[d]$ by $X_{i,j}$ and its $i$-th row vector by $X_i$. Further, let its $(i,j)$-minor be the $(d-1) \times (d-1)$ submatrix denoted by $X^{\setminus (i,j)}$ that results from $X$ by removing its $i$-th row and $j$-th column.
We denote by $I$ the $d\times d$ identity matrix and by $e_i$ its $i$-th column vector.

Our reduction involves the analysis of the \emph{determinant} $\det(X)$ and the \emph{adjugate} matrix $\adj(X)$, whose entries are defined as $\adj(X)_{i,j} = (-1)^{i+j}\cdot \det(X^{\setminus (j,i)})$.
We summarize certain well-known facts about these quantities and their relationships.
\begin{fact} 
{Let $X \in \mathbb{Z}_N^{d \times d}$ with linearly independent rows and columns. We remind of the following facts from linear algebra:}

\begin{enumerate}
\newcounter{foo}
\setcounter{foo}{1}
  \item $\det(X) \in \mathbb{Z}$ \label{F1:detinteger}
  \item $1\leq |\det(X)| \leq d^{d/2} N^{d}$ \label{F2:Hadamard}
  \item $X^{-1} = \det(X)^{-1}\adj(X)$ \label{F3:invdetadj}
  \item $\adj(X) \in \mathbb{Z}^{d\times d}$ \label{F4:adjinteger}
  \item Let $u,v\in \mathbb{R}^d$. If $X$ is invertible and $1+u^TX^{-1}v \neq 0$, then 
  \[ 
    X^{-1} - (X+u^Tv)^{-1} = \frac{X^{-1}uv^TX^{-1}}{1+u^TX^{-1}v}\,.
  \] \label{F5:shermanmorrison}
\end{enumerate}
\end{fact}
\begin{proof}
    By Leibnitz' formula, also known as the \emph{big formula} \cite{Strang2021}, the determinant of $X$ can be written as a signed sum over products of the entries of $X$. Since each $X_{i,j}\in\mathbb{Z}_N\subset \mathbb{Z}$, these products are again in $\mathbb{Z}$ and signed sums of the products are again integers in $\mathbb{Z}$. We thus get Fact~\ref{F1:detinteger}.

    Next we have by linear independence, that $\det(X)\neq 0$. From this and Fact~\ref{F1:detinteger} it follows that $1\leq |\det(X)|$. Hadamard's determinant bound \cite{Garling2007} together with $X_{i,j}\in\mathbb{Z}_N$ yields that 
    \[
        1\leq |\det(X)|\leq \prod_{j\in [d]}\left({\sum_{i\in[d]} X_{i,j}^2}\right)^{1/2} \leq \prod_{j\in [d]} ( dN^{2})^{1/2} = d^{d/2}N^d\,.
    \]
    We thus conclude Fact~\ref{F2:Hadamard}.

    Applying the Laplace expansion of the determinant, also known as \emph{cofactor formula} \cite{Strang2021}, to every row of $X$ it follows that $X\adj(X) = \det(X)I$. We recall that by linear independence, $X$ is invertible. Left multiplication with $X^{-1}$ thus yields $\adj(X)=X^{-1}X\adj(X) = X^{-1}\det(X)I$. By Fact~\ref{F2:Hadamard}, we have $\det(X)\neq 0$ and thus dividing both sides by $\det(X)$ concludes Fact~\ref{F3:invdetadj}, i.e., $X^{-1} = \det(X)^{-1}\adj(X)$.

    Note that each $(i,j)$-minor of $X\in\mathbb{Z}_N^{d \times d}$ for $i,j\in[d]$ is a submatrix $X^{\setminus (i,j)}\in \mathbb{Z}_N^{(d-1) \times (d-1)}$. Then Fact~\ref{F4:adjinteger} follows from Fact~\ref{F1:detinteger} since then $\adj(X)_{i,j} = (-1)^{i+j}\cdot \det(X^{\setminus (j,i)}) \in \mathbb{Z}$.

    Fact~\ref{F5:shermanmorrison} is a rearrangement of the Sherman-Morrison formula \cite{Strang2021} to isolate the remainder term of the inverse under a rank-$1$ update $X^{-1} - (X+u^Tv)^{-1}$.
\end{proof}

We will also need a bound on the coordinates of the inverse matrix multiplied by an arbitrary vector in the unit $\ell_\infty$ ball. 
To this end, we show that if $X$ has integer coordinates bounded between $-N$ and $N$, then the coordinates in the following inverse term can be controlled in terms of $N$ and $d$.

\begin{lemma}\label{lem:H-det}
   Let $\lambda \in [-1,1]^d$ and $X \in \mathbb{Z}_N^{d \times d}$.  Then for each coordinate $i \in [d]$ we have that
   \[  
      |(X^{-1} \lambda)_i| \leq \frac{1}{|\det(X)|} \cdot d^{(d+1)/2} N^{d-1}.
   \]    
\end{lemma}
\begin{proof}
    By Hoelder's inequality, the fact that $\|\lambda\|_\infty \leq 1$, and using Hadamard's determinant bound (Fact~\ref{F2:Hadamard}), we have that
    \begin{align*}
        |(X^{-1}\lambda)_i| 
        &\overset{Fact~\ref{F3:invdetadj}}{=} |(\det(X)^{-1} \adj(X)\lambda)_i|
        = \frac{1}{|\det(X)|} \cdot \left|(\adj(X)\lambda)_i\right| \nonumber \\
        &\!\!\overset{Hoelder}{\leq} \frac{1}{|\det(X)|} \cdot \|\adj(X)_i\|_1\|\lambda\|_\infty
        \leq \frac{1}{|\det(X)|} \cdot \sum_{j=1}^d \left|\det(X^{\setminus(j,i)})\right| \nonumber \\
        &\!\overset{Fact~\ref{F2:Hadamard}}{\leq} \frac{1}{|\det(X)|} \cdot d \cdot (d-1)^{(d-1)/2} N^{d-1}
        \leq \frac{1}{|\det(X)|} \cdot d^{(d+1)/2} N^{d-1}\, . \qedhere
    \end{align*}
\end{proof}

\subsubsection{Main Problem Reduction}
We are now ready to prove the main result of this section.

\begin{theorem}\label{thm:affinedegeneracy}
Assume that there exists no algorithm for {\rm \textsc{AffineDegeneracy}} on $n$ points with runtime $O(n^{d-c})$ for any $c >0$. Then there exist sets $R,B \subset \R^d$, each of size $n/2$, for which there is no algorithm that solves either 
\begin{itemize}
    \item \MaxH in time $O((n/2)^{d-c})$, or
    \item $\eps$-\MaxH in time $O(1/\eps^{d-c})$.
\end{itemize}
\end{theorem}

\begin{proof}
    Let $X_n=\{x_1 , \ldots, x_n\} \subset \mathbb{Z}_N^d$ be an instance of \textsc{AffineDegeneracy}.
    We construct an instance of \textsc{MaxHalfspace} by splitting each point into two points.

    \subparagraph*{The Construction.}
    For each $i\in[n]$ we create one point $y_{i}=x_i+\gamma e_1 \in R$ and one point $z_i=x_i-\gamma e_1 \in B $ where $\gamma$ is sufficiently small (yet to be determined).
    Now, we claim that there exists a solution with $\mu_R |R|-\mu_B |B|\geq d+1 $ if and only if there exists a set $X' \subseteq X_n$ of $|X'|=d+1$ points that lie on a common $d$-dimensional hyperplane. 

    First note that if there exists a set $X' \subseteq X_n$ of $|X'|=d+1$ points that lie on a common $d$-dimensional hyperplane $h$, then this hyperplane is a solution to \textsc{MaxHalfspace} with $\mu_R |R|-\mu_B |B|\geq d+1$ because for each $x\in X'$, the corresponding point $y\in R$ lies above $h$, whereas $z\in B$ lies below $h$.

    For the other direction, assume that there exists a hyperplane $h$ with $\mu_R |R|-\mu_B |B|\geq d+1$.
    Then there must be a set $S \subseteq [n]$ of $|S|=d+1$ indices such that for any $i \in S$ this hyperplane crosses $x_i'=x_i+\lambda_i \gamma e_1$ for some $\lambda_i \in [-1, 1]$, i.e., it passes between $y_i\in R$ and $z_i\in B$ corresponding to $x_i$. Assume w.l.o.g. that $S=[d+1]$. It remains to show that $x_1, \ldots, x_{d+1}$ lie on a common hyperplane.
    The key will be leveraging that the coordinates of input points are integers to ensure that the solution to a satisfying halfspace must have $x_{d+1}$ very close. Setting $\gamma$ sufficiently small will imply the distance is actually $0$.

    \subparagraph*{Reframing the Problem in Linear Algebra.}
    We assume w.l.o.g. that $x_1, \ldots, x_d$ and $x_1', \ldots, x_d'$ are linearly independent, {otherwise the proof continues verbatim in the largest dimension $d'<d$ for which the assumption is true, without affecting the hardness in $d$}. Then $x_1', \ldots, x_d'$ lie on a common hyperplane and in particular there exists $b' \in \mathbb{R}^d$ such that for all $i \in [d]$, it holds that $\langle b', x_i'\rangle -1 =0 $.
    We set $X \in \mathbb{Z}^{d\times d}$ to be matrix whose $i$-th row equals $x_i^T$, for $i\in[d]$, and $X' \in \mathbb{R}^{d\times d}$ to be matrix whose $i$-th row equals $x_i'^T$, for $i\in[d]$. 
    Let $\boldsymbol{1} \in \mathbb{R}^d$ be the vector whose coordinates are all equal to one. 
    Then we have that $X' b'= \boldsymbol{1}$, which implies $ b'=X'^{-1}\boldsymbol{1}$. Similarly, we set $b=X^{-1}\boldsymbol{1}$.
    
    Now, we show that $x_1, \ldots, x_{d+1}$ lie on a common hyperplane by proving the following two claims, which clearly imply that $\langle b,x_{d+1} \rangle-1=0$:
    \begin{enumerate}
    \item[(i)] we have that $\langle b,x_{d+1} \rangle- 1$ is a multiple of $1/N_0$, for some integer $N_0\geq 1$, and
    \item[(ii)] it holds that
    \[
        |\langle b,x_{d+1} \rangle- 1|=|\langle b,x_{d+1} \rangle- \langle b',x'_{d+1} \rangle + \langle b',x'_{d+1} \rangle - 1 | =|\langle b,x_{d+1} \rangle-\langle b',x'_{d+1} \rangle| < 1/N_0.
    \]
    \end{enumerate}
    
    \subparagraph*{(i) Showing Integrality.}
    For the first claim (i), note that $\langle b,x_{d+1} \rangle = (X^{-1}\boldsymbol{1})^T x_{d+1} $. Since $X$ comprises only integer entries, $\det(X)$ is also an integer by Fact~\ref{F1:detinteger}. By Fact~\ref{F2:Hadamard}, we also have that $|\det(X)|\geq 1$.
    We know from Fact~\ref{F3:invdetadj} that $X^{-1} = \adj(X)/\det(X)$ and the entries of $\adj(X)$ are again integers by Fact~\ref{F4:adjinteger}. Finally, each coordinate $i\in [d]$ of $\adj(X)\boldsymbol{1}$ is a sum over the integer entries of $\adj(X)_i$. The vector $x_{d+1}$ again consists only of integers. Taking their dot product thus retains the property of being integers, i.e., $(\adj(X)\boldsymbol{1})^Tx_{d+1} \in \mathbb{Z}$.
    
    We set $N_0 \coloneqq |\det(X)|$ and conclude that there exists $k\in \mathbb{Z}$ such that $\langle b,x_{d+1} \rangle -1 = \frac{k}{N_0}$. 

    \subparagraph*{(ii) Controlling the Deviations.}
    For the second claim (ii), let $R=X^{-1} - X'^{-1}$. 
    To bound the remainder term $R$, we first express $X'$ as a rank-$1$ update of $X$ as follows. Let $u=\gamma \lambda$, and let $v = e_1$. Then $X'^{-1} = (X + uv^T)^{-1}$. 
    Plugging in the definitions of $u$ and $v$, the Sherman-Morrison formula yields
    \begin{align*}
        R=X^{-1} - X'^{-1} \overset{Fact~\ref{F5:shermanmorrison}}{=} \frac{X^{-1}uv^TX^{-1}}{1+v^TX^{-1}u}
                           = \frac{X^{-1}\lambda\gamma e_1^TX^{-1}}{1+e_1^TX^{-1}\gamma \lambda}
                           = \frac{X^{-1}\lambda\gamma e_1^TX^{-1}}{1+(X^{-1}\lambda)^T\gamma e_1}\,.
    \end{align*}
    We will see below that the denominator is non-zero, so the remainder term is well-defined.
    Now, recall that $X^{-1}\boldsymbol{1} = b$. Similarly, we set $\beta\coloneqq X^{-1}\lambda$. Using these substitutions, we get
    \begin{align}\label{eq:R1}
        R\boldsymbol{1} 
        = (X^{-1} - X'^{-1})\boldsymbol{1}
        = \left(\frac{X^{-1}\lambda\gamma e_1^TX^{-1}\boldsymbol{1}}{1+(X^{-1}\lambda)^T\gamma e_1}\right)
        = \left(\frac{\beta\gamma e_1^Tb}{1+\beta^T\gamma e_1}\right)
        = \left(\frac{\gamma b_1}{1+\gamma \beta_1} \right) \beta\,.
    \end{align}
    We choose $\gamma \coloneqq (6d^{d+2}N^{2d})^{-1}>0$. 

    In particular, with our coordinate bound (\Cref{lem:H-det}) this implies that $\gamma |\beta_1| \ll 1/2$, such that the denominator in the remainder term of \Cref{eq:R1} is non-zero and bounded in the interval $[1/2,3/2]$.
    By Lemma~\ref{lem:H-det}, we can bound the coordinates of $\beta$ as $|\beta_i| = |(X^{-1} \lambda)_i| \leq \frac{1}{|\det(X)|} \cdot d^{(d+1)/2} N^{d-1}$. Since $\boldsymbol{1} \in [-1,1]^d$ the same bound holds for all coordinates $|b_i|$ of $b$.  
    These arguments yield the following bound
    \begin{align}\label{eq:helper1}
        \left|(R\boldsymbol{1})^T x_{d+1} \right| 
        &\overset{Eq.~\eqref{eq:R1}}{=} \left|\left(\frac{\gamma b_1}{1+\gamma \beta_1} \right) \beta^T x_{d+1} \right|
        \overset{Hoelder}{\leq} \gamma \left|\left(\frac{1}{1+\gamma \beta_1} \right)\right| |b_1| \|\beta\|_\infty \|x_{d+1}\|_1 \nonumber\\
        &\overset{Lem.~\ref{lem:H-det}}{\leq} \gamma\left|\left(\frac{1}{1+\gamma \beta_1} \right)\right| \cdot \frac{1}{|\det(X)|^2} \cdot d^{d+1} N^{2d-2} \cdot dN \nonumber\\
        &\quad\!<\quad 2\gamma \cdot\frac{1}{|\det(X)|} \cdot d^{d+2} N^{2d-1}
        < \frac{1}{3\cdot|\det(X)|} \,.
    \end{align}
    With almost the same calculation, we also have that
    \begin{align}\label{eq:helper2}
        \left|(R\boldsymbol{1})^T e_{1} \right| 
        &\overset{Eq.~\eqref{eq:R1}}{=} \left|\left(\frac{\gamma b_1}{1+\gamma \beta_1} \right) \beta_1 \right|
        \leq \gamma \left|\left(\frac{1}{1+\gamma \beta_1} \right)\right| |b_1| |\beta_1| \nonumber\\
        &\overset{Lem.~\ref{lem:H-det}}{\leq} \gamma\left|\left(\frac{1}{1+\gamma \beta_1} \right)\right| \cdot \frac{1}{|\det(X)|^2} \cdot d^{d+1} N^{2d-2} \nonumber\\
        &\quad\!<\quad 2\gamma \cdot\frac{1}{|\det(X)|} \cdot d^{d+1} N^{2d-2}
        < \frac{1}{3\cdot|\det(X)|} \,.
    \end{align}
    Now, by combining the triangle inequality, \Cref{eq:helper2}, Lemma \ref{lem:H-det}, and our choice of $\gamma =(6d^{d+2}N^{2d})^{-1} \leq 1$, we get that 
    \begin{align}\label{eq:helper3}
        \gamma\left|(R\boldsymbol{1})^T e_1 - (X^{-1}\boldsymbol{1})^Te_1 \right|
        &\leq \gamma|(R\boldsymbol{1})^T e_1| + \gamma|b_1| \nonumber \\
        &< \frac{1}{3\cdot|\det(X)|} + \frac{1}{6\cdot|\det(X)|} = \frac{1}{2\cdot|\det(X)|}\,.
    \end{align}
    We conclude the second claim (ii) as follows
    \begin{align*}
        &| \langle b, x_{d+1} \rangle-\langle b', x'_{d+1} \rangle | \\
        &=|(X^{-1}\boldsymbol{1})^T x_{d+1}- (X'^{-1}\boldsymbol{1})^T x'_{d+1} | \\
        &=|(X^{-1}\boldsymbol{1})^T x_{d+1}- (X'^{-1}\boldsymbol{1}-X^{-1}\boldsymbol{1}+X^{-1}\boldsymbol{1})^T x'_{d+1} | \\
        &=|(X^{-1}\boldsymbol{1})^T x_{d+1}- (X^{-1}\boldsymbol{1})^T x'_{d+1}+ (X^{-1}\boldsymbol{1}-X'^{-1}\boldsymbol{1})^T x'_{d+1} | \\
        &=|(X^{-1}\boldsymbol{1})^T x_{d+1} -(X^{-1}\boldsymbol{1})^T x_{d+1} - \gamma \lambda_{d+1} (X^{-1}\boldsymbol{1})^Te_1 + (R\boldsymbol{1})^T x'_{d+1} | \\
        &=|(R\boldsymbol{1})^T x_{d+1} + \gamma \lambda_{d+1}(R\boldsymbol{1})^T e_1 - \gamma \lambda_{d+1}(X^{-1}\boldsymbol{1})^Te_1 | \\
        &\leq |(R\boldsymbol{1})^T x_{d+1}| + \gamma |\lambda_{d+1}| |(R\boldsymbol{1})^T e_1 - (X^{-1}\boldsymbol{1})^Te_1 | \\
        &< \frac{1}{3\cdot|\det(X)|} + \frac{1}{2\cdot|\det(X)|} < \frac{1}{|\det(X)|}\,,
    \end{align*}
    where the triangle inequality is followed by using the fact that $|\lambda_{d+1}|\leq 1$ and applying \Cref{eq:helper1,eq:helper3}.

    \subparagraph*{Final Reductions.}
    Finally, we set $m = n/2$ and $\eps=1/n$.
    Then, any algorithm that solves \MaxH on $R$ and $B$ in time $O(m^{d-c}) = O((n/2)^{d-c})$ time or 
    $\eps$-\MaxH in $O(\eps^{-(d-c)})$ time can also be used to solve \textsc{AffineDegeneracy} in time $O(n^{d-c})$.
\end{proof}

\bibliography{references}

\end{document}